\documentclass[12pt,a4paper]{article}

\usepackage[utf8]{inputenc}
\usepackage[T1]{fontenc}
\usepackage{amsmath,amssymb,amsthm}
\usepackage{geometry}
\usepackage{setspace}
\usepackage{natbib}
\usepackage{booktabs}
\usepackage{tabularx}
\usepackage{graphicx}
\usepackage{hyperref}
\usepackage{enumitem}
\usepackage{float}
\usepackage{caption}
\usepackage{xcolor}
\usepackage{array}
\usepackage{multirow}
\usepackage[protrusion=true,expansion=false]{microtype}  % Protrusion only — safe with CM fonts

\hypersetup{colorlinks=true,linkcolor=blue!60!black,citecolor=blue!60!black,urlcolor=blue!60!black}

\renewcommand{\arraystretch}{1.15}           % Slightly more breathing room in all tables

\newtheorem{proposition}{Proposition}
\newtheorem{corollary}{Corollary}
\newtheorem{definition}{Definition}
\theoremstyle{remark}
\newtheorem*{remark}{Remark}

\title{\textbf{The Data Hydration Gap: A Formal Model of Underinvestment in General-Purpose Data Products Under Decentralized Governance}}
\author{Gaston Besanson\thanks{Universidad Torcuato Di Tella.}}
\date{Working Paper --- March 2026}

\begin{document}

\maketitle

\begin{abstract}
\noindent
When organizations decentralize data product ownership---as in the data mesh paradigm---each domain team optimizes for its immediate analytical needs, underinvesting in the cross-domain \emph{generality} that enables organization-wide reuse. We formalize this as a simultaneous-move game in which $N$ domains choose quality ($q$) and generality ($g$); generality creates positive externalities but is privately costly. The Nash equilibrium generality gap is increasing in the number of domains and the value of cross-domain analytics; under plausible parameter configurations, a corner solution obtains in which no reusable ``silver layer'' emerges organically---a condition we term the \emph{data mesh trap}. Technical debt from narrow products grows quadratically in $N$. An illustrative calibration suggests non-trivial organizational welfare losses under plausible enterprise parameters. We derive within-model conditions under which centralized, federated, and hybrid governance regimes dominate, and we identify the information asymmetries and transaction costs that complicate implementation. The model provides a formal foundation for empirical research on decentralized data governance.


\end{abstract}

\newpage
\tableofcontents
\newpage

%=============================================================
% 1. INTRODUCTION
%=============================================================
\section{Introduction}\label{sec:intro}

The data mesh paradigm \citep{Dehghani2019, Dehghani2022} advocates domain-oriented data ownership: individual business domains create and maintain their own ``data products,'' promising contextual understanding, rapid iteration, and freedom from centralized bottlenecks. Practitioner adoption has been rapid and the framework has attracted growing IS research attention \citep{Santos2023, Driessen2023}.

This paper develops a formal incentive analysis of a structural problem with decentralized data product governance: the \emph{generality underinvestment problem}. When domain teams independently structure their data products, each optimizes for immediate analytical use cases. Creating general-purpose products---with standardized schemas, cross-domain semantic alignment, and comprehensive documentation---requires investment whose benefits accrue primarily to other domains, generating a positive externality that decentralized governance fails to internalize.

The problem has the structure of a public goods game. General-purpose data products exhibit non-rivalry and partial excludability \citep{Samuelson1954}. The resulting underinvestment parallels classic externality results \citep{Pigou1920}, but the intra-organizational setting introduces distinctive features: governance can deploy authority-based mechanisms (mandates, budgets, subsidies), but these face transaction costs related to political capital, autonomy norms, and information asymmetry \citep{Williamson1985, Gulati2012}.

We formalize this as a simultaneous-move game in which $N$ domains choose quality ($q$) and generality ($g$). Our central result (Proposition~\ref{prop:underinvestment}) establishes that each domain invests in generality only where private marginal benefits equal private marginal costs---ignoring the cross-domain benefits its generality creates. The gap grows with $N$ and cross-domain analytics value. We identify conditions under which a corner solution obtains ($g = 0$ for all $i$), which we term the \emph{data mesh trap} (Corollary~\ref{cor:trap}), defined as $\alpha_i \beta < \kappa / q_i^*$ for all $i$.

The paper contributes to the IS and data governance literatures in three ways:

\begin{enumerate}[leftmargin=2em]
    \item We develop a formal incentive analysis of decentralized internal data-product governance, characterizing the externality structure and equilibrium underinvestment. This provides microeconomic foundations for a phenomenon documented in practitioner case studies \citep{Machado2022, Buger2023} but not previously formalized.
    
    \item We provide an illustrative calibration to enterprise data architectures using observable proxies, illustrating that the mechanism can generate non-trivial welfare losses under plausible parameter configurations. These are explicitly model-implied scenario outputs, not empirical measurements; we specify the identification challenges that future empirical work must address.
    
    \item We analyze alternative governance regimes---centralized hydration, federated governance with incentive alignment, and hybrid approaches---deriving within-model conditions under which each dominates. The analysis suggests that the theoretical optimality of Pigouvian subsidies must be tempered by the transaction costs, information asymmetries, and organizational frictions that complicate implementation.
\end{enumerate}

Section~\ref{sec:lit} positions the model within intersecting literatures. Section~\ref{sec:model} develops the formal model with explicit scope conditions. Section~\ref{sec:equilibrium} characterizes the Nash equilibrium. Section~\ref{sec:welfare} analyzes welfare and technical debt. Section~\ref{sec:calibration} presents an illustrative calibration. Section~\ref{sec:governance} examines governance regimes. Section~\ref{sec:predictions} derives testable predictions. Section~\ref{sec:discussion} discusses limitations and concludes.

%=============================================================
% 2. RELATED LITERATURE AND THEORETICAL FOUNDATIONS
%=============================================================
\section{Related Literature and Theoretical Foundations}\label{sec:lit}

This paper enters the intersection of five research streams: (a)~data mesh and data architecture, (b)~externalities and public goods, (c)~platform economics and two-sided markets, (d)~architectural technical debt, and (e)~organizational design and governance of digital infrastructure. We identify what each contributes, where each falls short, and how our formal treatment advances the conversation.

\subsection{Data Mesh and Data Architecture}\label{sec:lit:mesh}

The data mesh concept \citep{Dehghani2019, Dehghani2022} rests on four principles: domain-oriented ownership, data as a product, self-serve data infrastructure, and federated computational governance. The practitioner and gray literature has grown substantially \citep[for systematic reviews, see][]{Santos2023, Driessen2023}, and early case-study research documents both adoption patterns and recurring failure modes \citep{Machado2022, Buger2023}. However, formal economic analysis of the incentive structures underlying data mesh outcomes remains limited. Most existing treatments document \emph{what} happens when organizations adopt data mesh, without modeling \emph{why} specific failures arise from the microeconomic incentives facing domain teams.

The silver layer concept originates in the medallion architecture \citep{Databricks2021}, distinguishing bronze (raw ingestion), silver (cleaned and conformed), and gold (business-level aggregates) data layers. Our model focuses on the silver layer as the critical locus of generality decisions. We define \emph{data hydration} as the systematic process of validating, cleansing, standardizing, and enriching data from raw source systems into reusable silver-layer data products. Within the model, centralized hydration operates as a public infrastructure subsidy: either a centralized reduction in the marginal generalization cost ($\gamma_g$) through shared platform tooling, or a centralized assumption of the fixed standardization cost ($\kappa$) on behalf of domains.

\subsection{Externalities, Public Goods, and Intra-Organizational Market Failures}\label{sec:lit:externalities}

The generality underinvestment problem has the analytical structure of a public goods game with positive externalities \citep{Samuelson1954, Pigou1920}. Our model adapts these frameworks to intra-organizational data products, where generality exhibits non-rivalry and partial excludability. Unlike traditional public goods settings, however, producers and consumers operate within a single organizational boundary, creating governance possibilities (mandates, budgets, subsidies) and constraints (political capital, autonomy norms) that differ from the public sector \citep{Williamson1985}.

The organizational economics literature on internal markets and transfer pricing \citep{Eccles1985, Holmstrom1991} is directly relevant: the problem of incentivizing domain teams to produce general data products mirrors the challenge of setting internal prices that balance local optimization against firm-wide efficiency. The distinction between authority-based and incentive-based coordination \citep{Gulati2012, Puranam2014} maps onto our governance regime comparison: centralized hydration is authority-based, while federated governance with Pigouvian subsidies is incentive-based.

\subsection{Platform Economics and Internal Data Markets}\label{sec:lit:platform}

Our model connects to the platform economics literature \citep{RochetTirole2003, Armstrong2006} through the observation that internal data platforms mediate between data producers (domains) and consumers (analytics teams). Recent work on data bottlenecks in digital platforms---where information asymmetry and exclusive data control lead to market failures \citep{Duch-Brown2017}---provides an external analogy. In the intra-organizational context, domain teams function as localized micro-monopolies over their operational data, and the internal data market suffers from analogous externality failures.

The emerging literature on inter-organizational data sharing \citep{Gelhaar2021, Oliveira2019} highlights tensions between monetization incentives and combinatorial value creation---a dynamic that parallels our setting. The mathematical dynamics of our propositions hold regardless of whether interacting nodes are internal departments or independent enterprises, suggesting generalizability to sovereign data spaces (e.g., Gaia-X), though such extension would require modeling inter-firm contracting frictions.

\subsection{Architectural Technical Debt as Economic Externality}\label{sec:lit:atd}

The technical debt metaphor \citep{Cunningham1992}, elaborated in influential typologies \citep{Kruchten2012}, distinguishes code-level debt from architectural technical debt (ATD). Unlike code debt, ATD involves fundamental structural design choices that are difficult to resolve without systemic rework \citep{Verdecchia2021, Besker2018}. Our Proposition~\ref{prop:debt} provides a formal economic foundation for ATD accumulation in decentralized data systems, demonstrating quadratic scaling in the number of domains. We reframe ATD not merely as an engineering oversight, but as an \emph{inevitable economic externality} of uncoordinated decentralized system design---connecting the software engineering literature to the economics of organizational governance.

\subsection{Organizational Design and Digital Infrastructure Governance}\label{sec:lit:orgdesign}

The IS literature on IT governance \citep{Weill2004, Tiwana2014} has long examined the tension between centralized control and decentralized autonomy over digital resources. \citet{Ceccagnoli2012} analyzes how governance structures shape complementor investment. Our contribution is distinct: we model the \emph{specific externality structure} of data product generality under decentralized ownership, deriving the conditions under which underinvestment arises rather than prescribing governance arrangements ex ante.

\subsection{Literature Gap and Positioning}\label{sec:lit:gap}

The intersection of these literatures reveals a clear gap. Existing data mesh research documents architectural patterns and prescriptions; public-goods theory explains underprovision abstractly; platform economics models mediated exchange across firm boundaries; ATD scholarship documents debt consequences through cases; and organizational design research examines authority--autonomy trade-offs at a general level. We contribute a tractable formal model that integrates these perspectives to explain why, under specified incentive conditions, decentralized data governance may fail to produce reusable data products, and what governance mechanisms can in principle correct the underlying externality.

%=============================================================
% 3. THE MODEL
%=============================================================
\section{The Model}\label{sec:model}

\subsection{Setup}\label{sec:model:setup}

We model a data organization comprising $N$ domain teams ($i = 1, \ldots, N$), $M$ cross-domain consumers ($j = 1, \ldots, M$), and an optional central data platform. Each domain $i$ creates data products characterized by two attributes:
\begin{itemize}[leftmargin=2em]
    \item \textbf{Quality} $q_i \in [0,1]$: accuracy, completeness, freshness, and documentation.
    \item \textbf{Generality} $g_i \in [0,1]$: degree to which the product serves cross-domain needs.
\end{itemize}

\begin{remark}[Construct clarification: generality as reduced-form reusability]
In practice, cross-domain reusability encompasses schema standardization, semantic alignment, documentation completeness, interface stability, and metadata governance. For mathematical tractability, we collapse these into a single continuous variable $g_i$ serving as a reduced-form measure of cross-domain reusability---formally, the inverse of the integration cost experienced by out-of-domain consumers. This reduced-form approach allows the formal results to hold without over-claiming empirical exactitude regarding how standardization is technically achieved. The multidimensionality of the underlying construct is a limitation we return to in Section~\ref{sec:discussion:limits}.
\end{remark}

\subsection{Costs}\label{sec:model:costs}

Domain $i$'s cost of creating a data product is:
\begin{equation}\label{eq:cost}
    C_i(q_i, g_i) = \frac{\gamma_q}{2}\, q_i^2 + \frac{\gamma_g}{2}\, g_i^2 \cdot q_i + \kappa \cdot g_i
\end{equation}

The first term captures quadratic quality costs (diminishing returns). The second is an interaction: making high-quality data general is harder than making low-quality data general, because standardizing accurate data requires more rigorous semantic reconciliation. The third represents fixed generalization costs (schema governance overhead, cross-domain coordination, semantic standardization).

\subsection{Benefits}\label{sec:model:benefits}

Domain $i$ receives benefits from two sources. Direct benefit from its own analytics:
\begin{equation}\label{eq:own}
    B_i^{\text{own}}(q_i, g_i) = \alpha_i \cdot q_i \cdot (1 + \beta \cdot g_i)
\end{equation}
where $\alpha_i$ is domain $i$'s analytics value and $\beta \geq 0$ captures the synergy between quality and generality---reflecting that generalization often forces better discipline (catching errors, improving documentation).

Cross-domain benefit from accessing other domains' products:
\begin{equation}\label{eq:cross}
    B_i^{\text{cross}}(\{q_j, g_j\}_{j \neq i}) = \lambda \sum_{j \neq i} q_j \cdot g_j
\end{equation}
where $\lambda > 0$ captures the value of cross-domain data access. Only general products ($g_j > 0$) are usable by domain $i$---narrow products are too specific to provide value elsewhere.

Cross-domain consumer $j$'s utility:
\begin{equation}\label{eq:consumer}
    U_j = \sum_{i=1}^{N} \omega_{ij} \cdot q_i \cdot g_i - S \cdot \sum_{i=1}^{N} (1 - g_i)
\end{equation}
where $\omega_{ij}$ is consumer $j$'s weight on domain $i$'s data and $S$ is the switching/integration cost.

\subsection{Timing, Information Structure, and Assumptions}\label{sec:model:timing}

The game proceeds in two stages: (1)~each domain $i$ simultaneously chooses $(q_i, g_i)$; (2)~cross-domain consumers access products, benefits realize, payoffs are determined.

The baseline model rests on three key assumptions, each of which we revisit as explicit limitations in Section~\ref{sec:discussion:limits}:

\begin{enumerate}[leftmargin=2em]
    \item \textbf{Complete information.} All domains observe the full parameter vector. In practice, the true quality and generality of data products may be unverifiable ex ante, creating moral hazard (see Section~\ref{sec:discussion:limits}, ``Information assumptions'').
    \item \textbf{Symmetric domains.} All domains share identical $\alpha_i$, $\gamma_q$, $\gamma_g$, $\kappa$. Real organizations are inherently asymmetric; we explore this in Appendix~\ref{app:heterogeneous} and discuss it as a limitation in Section~\ref{sec:discussion:limits}, ``Symmetric domains.''
    \item \textbf{Static game.} The core equilibrium is one-shot. We introduce a multi-period extension for technical debt in Section~\ref{sec:welfare:debt}, but a fully dynamic model with learning and reputation is left for future work (Section~\ref{sec:discussion:limits}, ``Static model'').
\end{enumerate}

%=============================================================
% 4. EQUILIBRIUM ANALYSIS
%=============================================================
\section{Equilibrium Analysis}\label{sec:equilibrium}

\subsection{Domain Optimization}\label{sec:eq:optimization}

Domain $i$ maximizes:
\begin{equation}\label{eq:profit}
    \pi_i = B_i^{\text{own}}(q_i, g_i) + B_i^{\text{cross}}(\{q_j, g_j\}_{j \neq i}) - C_i(q_i, g_i)
\end{equation}

Taking the first-order condition with respect to $g_i$:
\begin{equation}\label{eq:foc}
    \frac{\partial \pi_i}{\partial g_i} = \alpha_i \beta q_i - \gamma_g g_i q_i - \kappa = 0
\end{equation}

Solving:
\begin{equation}\label{eq:gstar}
    g_i^* = \max\left\{0, \; \frac{\alpha_i \beta - \kappa / q_i^*}{\gamma_g}\right\}
\end{equation}

Domain $i$'s FOC for $g_i$ does not include the cross-domain benefit it provides to others: $\sum_{j \neq i} \lambda q_i$. This is the externality gap.

\begin{proposition}[Generality Underinvestment]\label{prop:underinvestment}
Within the model, in Nash equilibrium, each domain sets $g_i^{NE}$ strictly less than the socially optimal level $g_i^{SO}$. The generality gap is:
\begin{equation}\label{eq:gap}
    \Delta g_i = g_i^{SO} - g_i^{NE} = \frac{(N-1)\lambda}{\gamma_g}
\end{equation}
The gap is increasing in $N$ and $\lambda$, and decreasing in $\gamma_g$.
\end{proposition}

\begin{proof}
See Appendix~\ref{app:proof1}.
\end{proof}

\paragraph{Result.} The private first-order condition omits the benefit conferred on out-of-domain users. The gap expression $\frac{(N-1)\lambda}{\gamma_g}$ is exact within the model under the stated assumptions.

\paragraph{Economic intuition.} A domain captures the cost of standardization but not the full value it creates for other domains. As the organization grows (higher $N$) or cross-domain analytics become more valuable (higher $\lambda$), the wedge between private and social incentives widens. Each domain individually behaves rationally; the inefficiency arises from the interaction structure, not from irrationality.

\paragraph{Scope conditions.} Expression (\ref{eq:gap}) abstracts from the consumer term $M\bar{\omega}_i$; retaining it yields $\Delta g_i = \frac{(N-1)\lambda + M\bar{\omega}_i}{\gamma_g}$, which is weakly larger. The result assumes interior solutions at the social optimum; the corner case is Corollary~\ref{cor:trap}. Symmetry is relaxed in Appendix~\ref{app:heterogeneous}.

\subsection{Corner Solution: The Data Mesh Trap}\label{sec:eq:corner}

\begin{corollary}[Data Mesh Trap]\label{cor:trap}
When $\alpha_i \beta < \kappa / q_i^*$ for all domains, the Nash equilibrium has $g_i^{NE} = 0$ for all $i$. No silver layer emerges organically.
\end{corollary}

The corner solution obtains when private synergy benefits ($\alpha_i \beta$) fall below effective fixed standardization costs ($\kappa / q_i^*$). Within the model, this defines the \emph{data mesh trap}: the parameter region yielding zero generality despite positive social value.

\paragraph{Result.} The trap follows from the externality structure combined with fixed costs. No domain can unilaterally break the trap because generality costs are private while benefits are shared.

\paragraph{Economic intuition.} The condition is more likely when $\beta$ is low (generality does not improve the domain's own analytics), $\kappa$ is high, or $\alpha_i$ is low. Paradoxically, domains that care least about data quality are most likely to produce narrow products, even when their data is valuable to others.

\paragraph{Scope conditions.} The data mesh trap describes a corner solution within the model, not a universal diagnosis. Whether it describes a particular organization is an empirical question.

%=============================================================
% 5. WELFARE ANALYSIS AND TECHNICAL DEBT
%=============================================================
\section{Welfare Analysis and Technical Debt}\label{sec:welfare}

\subsection{Total Welfare}\label{sec:welfare:total}

Total welfare:
\begin{equation}\label{eq:welfare}
    W = \sum_{i=1}^{N} \pi_i + \sum_{j=1}^{M} U_j
\end{equation}

\begin{proposition}[Welfare Loss from Decentralization]\label{prop:welfareloss}
Within the model, the welfare loss from Nash equilibrium relative to the social optimum is:
\begin{equation}\label{eq:welfareloss}
    \Delta W = \sum_{i=1}^{N} \left[(N-1)\lambda + M\bar{\omega}_i\right] q_i^* \cdot \Delta g_i - \frac{\gamma_g}{2} \sum_{i=1}^{N} \left[(g_i^{SO})^2 - (g_i^{NE})^2\right] q_i^*
\end{equation}
Welfare loss is increasing in $N^2$ and in $\lambda$ and $M$.
\end{proposition}

\paragraph{Result.} Welfare loss has three components: cross-domain inefficiency, consumer harm (redundant custom pipelines), and duplicated effort.

\paragraph{Economic intuition.} The quadratic dependence on $N$ arises because externalities are pairwise: each additional domain creates new unrealized gains from trade with every existing domain.

\paragraph{Scope conditions.} The expression assumes the social planner can observe and set generality without additional cost. If the planner faces its own information or coordination costs, the achievable welfare improvement is smaller (see Section~\ref{sec:governance}).

\subsection{Technical Debt Accumulation}\label{sec:welfare:debt}

We extend the model to multiple periods. If domain $i$ needs data from domain $j$ in period~2 but $g_j = 0$, domain $i$ must build custom integration:
\begin{equation}\label{eq:td}
    TD_{ij} = \tau \cdot q_j \cdot (1 - g_j)
\end{equation}
Total organizational technical debt:
\begin{equation}\label{eq:tdtotal}
    TD_{\text{total}} = \tau \sum_{i=1}^{N} \sum_{j \neq i} q_j^* \cdot (1 - g_j^{NE}) \cdot P_{ij}
\end{equation}

\begin{proposition}[Technical Debt Cascade]\label{prop:debt}
Under Nash equilibrium with $g_i^{NE} = 0$ for all $i$, and with symmetric domains ($q_j = \bar{q}$, $P_{ij} = \bar{P}$):
\begin{equation}\label{eq:debtquad}
    TD_{\text{total}} = \tau \cdot \bar{q} \cdot N(N-1) \cdot \bar{P}
\end{equation}
Technical debt grows quadratically in $N$.
\end{proposition}

\begin{proof}
See Appendix~\ref{app:proof3}.
\end{proof}

\paragraph{Result.} Under the corner solution, every cross-domain data need generates integration debt. The $N(N-1)$ term is exact for symmetric domains.

\paragraph{Economic intuition.} With 5 domains there are 20 potential integration pairs; with 20 domains, 380. Each narrow product creates burden for every potential consumer. This formalizes what the ATD literature \citep{Kruchten2012, Verdecchia2021} has documented through qualitative case studies.

\paragraph{Scope conditions.} Quadratic scaling requires the corner solution ($g = 0$ for all domains). Under interior solutions, debt scales less than quadratically. The symmetry assumption ($q_j = \bar{q}$, $P_{ij} = \bar{P}$) is relaxed in Appendix~\ref{app:heterogeneous}.

%=============================================================
% 6. ILLUSTRATIVE CALIBRATION AND SENSITIVITY ANALYSIS
%=============================================================
\section{Illustrative Calibration and Sensitivity Analysis}\label{sec:calibration}

The purpose of this section is not statistical estimation but disciplined magnitude illustration under plausible enterprise parameter values. The welfare figures that follow are model-implied scenario outputs, not identified causal estimates. We classify inputs as observed, inferred, or illustrative and specify what would be needed for true empirical calibration.

\paragraph{What true calibration would require.} Empirical identification of the model's core parameters would demand: (a)~exogenous variation in governance regimes (e.g., post-merger integration events, regulatory mandates requiring data mapping) to identify causal effects on generality; (b)~panel data with domain-level fixed effects to address the endogeneity of autonomy and generality; and (c)~validated measurement instruments for cross-domain data value ($\lambda$), which is likely endogenous to the existing level of generality.

\subsection{Parameter Identification}\label{sec:cal:params}

Table~\ref{tab:params} maps model parameters to observable proxies.

\begin{table}[H]
\centering
\caption{Observable Proxies for Model Parameters}\label{tab:params}
\footnotesize
\begin{tabularx}{\textwidth}{@{}l l l X l X@{}}
\toprule
\textbf{Param.} & \textbf{Meaning} & \textbf{Source} & \textbf{Observable Proxy} & \textbf{Range} & \textbf{Ident.\ Challenge} \\
\midrule
$\alpha$ & Analytics value & Inferred & Team size, decision frequency & $0.3$--$0.8$ & Confounded with domain maturity \\[2pt]
$\beta$ & Gen.--quality synergy & Illustrative & Quality change post-standardization & $0.1$--$0.3$ & Requires pre/post design \\[2pt]
$\lambda$ & Cross-domain value & Inferred & Cross-domain request ratio & $0.2$--$0.6$ & Endogenous to $g$ \\[2pt]
$\gamma_g$ & Generalization cost & Observed & Std.\ hours / build hours & $0.3$--$0.5$ & Stack-dependent \\[2pt]
$\kappa$ & Fixed std.\ cost & Observed & Governance overhead (FTE) & $0.1$--$0.4$ & Organization-specific \\[2pt]
$N$ & Domains & Observed & Organizational structure & $5$--$20$ & Endogenous to design \\[2pt]
$S$ & Switching cost & Inferred & Pipeline build cost & \$10K--100K & Legacy system confound \\
\bottomrule
\end{tabularx}
\begin{minipage}{\textwidth}
\vspace{4pt}
\footnotesize \textit{Notes.} ``Source'' classifies each parameter as directly \textit{observed} from enterprise telemetry, \textit{inferred} from organizational proxies, or \textit{illustrative} (assumed for scenario analysis). See Section~\ref{sec:calibration} header for epistemological framing.
\end{minipage}
\end{table}

\subsection{Baseline Calibration}\label{sec:cal:baseline}

We calibrate to a representative large enterprise: 12 domain teams, $\alpha = 0.5$, $\beta = 0.15$, $\lambda = 0.4$, $\gamma_g = 0.4$, $\kappa = 0.25$, $q^* \approx 0.6$.

\begin{equation}\label{eq:calibration_ne}
    g_i^{NE} = \max\left\{0, \; \frac{0.075 - 0.417}{0.4}\right\} = 0
\end{equation}

The model predicts a corner solution. The social optimum involves:
\begin{equation}\label{eq:calibration_so}
    g_i^{SO} = \frac{0.075 + 4.4 - 0.417}{0.4} \approx 0.58
\end{equation}

\subsection{Sensitivity Analysis}\label{sec:cal:sensitivity}

Table~\ref{tab:sensitivity} assesses robustness of the corner solution across the parameter space.

\begin{table}[H]
\centering
\caption{Sensitivity of Nash Equilibrium to Parameter Variation}\label{tab:sensitivity}
\footnotesize
\begin{tabularx}{\textwidth}{@{}l l X X@{}}
\toprule
\textbf{Parameter} & \textbf{Range} & \textbf{Corner Breaks When} & \textbf{Impact on $\Delta g_i$} \\
\midrule
$\lambda$ (cross-domain value) & $0.1$--$0.9$ & Does not break corner; widens gap & Gap increases linearly in $\lambda$ \\[2pt]
$\kappa$ (fixed std.\ cost) & $0.05$--$0.5$ & $\kappa < \alpha\beta q^* \approx 0.045$ & Lower $\kappa$ enables $g > 0$ \\[2pt]
$\gamma_g$ (marg.\ gen.\ cost) & $0.2$--$0.8$ & Only if $\kappa$ also overcome & Gap decreases in $\gamma_g$ \\[2pt]
$N$ (domains) & $3$--$50$ & Does not break corner; widens gap & $g^{SO}$ increases linearly in $N$ \\[2pt]
$\beta$ (synergy) & $0.05$--$0.8$ & $\beta > \kappa / (\alpha q^*) \approx 0.83$ & Higher $\beta$ makes generality privately valuable \\
\bottomrule
\end{tabularx}
\begin{minipage}{\textwidth}
\vspace{4pt}
\footnotesize \textit{Notes.} Each row varies one parameter, holding others at baseline. Increasing $\lambda$ does not break the corner because $\lambda$ enters only the planner's objective, not the domain's---the essence of the externality.
\end{minipage}
\end{table}

The critical threshold is the relationship between the private synergy benefit ($\alpha\beta$) and the effective fixed cost ($\kappa / q^*$). The parameter that makes generality most socially valuable ($\lambda$) is precisely the one that private incentives ignore.

\subsection{Welfare Loss Estimation}\label{sec:cal:welfare}

We calculate \emph{illustrative} annual welfare losses per domain from three sources: duplicated engineering effort (approximately \$300K), integration overhead for custom pipelines (approximately \$250K), and data quality issues from inconsistent definitions (approximately \$200K). This yields roughly \$750K per domain annually. For a 12-domain organization, $\Delta W \approx$ \$9M; for enterprises with 20 or more domains, $\Delta W \approx$ \$20M.

These are order-of-magnitude figures under assumed parameter values. Actual losses depend on technology maturity, governance culture, data complexity, and factors not captured in the model. The figures should not be cited as empirical measurements.

%=============================================================
% 7. GOVERNANCE REGIMES
%=============================================================
\section{Governance Regimes}\label{sec:governance}

\subsection{Pure Data Mesh (Decentralized)}\label{sec:gov:mesh}

Under pure data mesh, each domain independently chooses $(q_i, g_i)$. As shown in Section~\ref{sec:equilibrium}, this yields systematic underinvestment. The regime is inefficient but imposes minimal coordination costs and preserves domain autonomy.

\subsection{Centralized Hydration}\label{sec:gov:central}

A central platform invests in silver-layer products, choosing organization-wide generality standard $G$:
\begin{equation}\label{eq:central}
    \max_G \;\; \sum_{i=1}^{N} \left[\lambda \sum_{j \neq i} q_j G \right] + \sum_{j=1}^{M} U_j(G) - \Gamma \cdot G^2 \cdot \bar{q}
\end{equation}
Centralized investment approaches the social optimum when $\Gamma$ is sufficiently low, reflecting economies of scale. The regime faces bottleneck risk, as centralized teams may lack domain context.

\subsection{Federated Governance with Incentive Alignment}\label{sec:gov:federated}

A Pigouvian subsidy of:
\begin{equation}\label{eq:subsidy}
    s_i = (N - 1)\lambda \cdot q_i
\end{equation}
exactly corrects the externality within the model.

\paragraph{Implementation caveats.} Calculating $s_i$ requires knowledge of $\lambda$ and observation of true generality---both may be unobservable or strategically misrepresented. Subsidies based on \emph{declared} rather than \emph{verified} generality create incentives to game the metric. The resulting moral hazard and transaction costs may render the theoretical optimum practically unworkable, leading organizations to rationally choose second-best solutions.

\subsection{Regime Comparison}\label{sec:gov:comparison}

Within the model, federated governance ranks highest on welfare (exact externality correction), followed by centralized hydration (effective if scale economies are sufficient), hybrid approaches (partial correction), and pure data mesh (no correction). This ranking is conditional on the model's assumptions---costless verifiability and absent implementation friction. When transaction costs and information asymmetries are material, the ranking may shift. Table~\ref{tab:regimes} summarizes.

\begin{table}[H]
\centering
\caption{Model-Implied Regime Comparison Under Baseline Parameters}\label{tab:regimes}
\footnotesize
\begin{tabularx}{\textwidth}{@{}X l l X X@{}}
\toprule
\textbf{Regime} & \textbf{Equil.\ $g$} & \textbf{Coord.\ cost} & \textbf{Relative welfare} & \textbf{Impl.\ friction} \\
\midrule
Pure Data Mesh & $g = 0$ & Minimal & Lowest (full welfare loss) & Minimal \\[2pt]
Centralized Hydration & $g = G^*$ & Mod.--high & Second-highest & High (bottleneck risk) \\[2pt]
Federated + Incentives & $g = g^{SO}$ & Low (if verif.) & Highest (within model) & Very high (info asymmetry) \\[2pt]
Hybrid (central silver) & $g \approx 0.7 g^{SO}$ & Moderate & Third-highest & Moderate \\
\bottomrule
\end{tabularx}
\begin{minipage}{\textwidth}
\vspace{4pt}
\footnotesize \textit{Notes.} Welfare ranking is derived from the model under baseline parameters (Section~\ref{sec:calibration}). ``Relative welfare'' describes ordinal ranking, not cardinal measurement. The welfare gap between decentralized and socially optimal outcomes is on the order of several million dollars annually for a 12--20 domain organization; see Section~\ref{sec:cal:welfare}. The ranking is conditional on the model's information assumptions; under substantial transaction costs, centralized or hybrid approaches may dominate.
\end{minipage}
\end{table}

%=============================================================
% 8. TESTABLE PREDICTIONS AND EMPIRICAL DESIGN
%=============================================================
\section{Testable Predictions and Empirical Design}\label{sec:predictions}

The model generates falsifiable predictions. For each, we specify mechanism, outcome variable, empirical proxy, unit of analysis, data source, identification challenge, and preferred design.

\begin{table}[H]
\centering
\caption{Empirical Research Design for Model Predictions}\label{tab:predictions}
\scriptsize
\renewcommand{\arraystretch}{1.3}
\begin{tabularx}{\textwidth}{@{}>{\bfseries}l X X X@{}}
\toprule
 & \textbf{Prediction 1} & \textbf{Prediction 2} & \textbf{Prediction 3} \\
\midrule
Mechanism & Private FOC omits cross-domain externality & Each narrow product creates $N{-}1$ integration liabilities & Centralized generality reduces per-consumer integration cost \\
Outcome var. & Data product generality & Cross-domain integration cost & Time-to-delivery for cross-domain analytics \\
Empirical proxy & Cross-domain API calls / total; metadata compliance & Point-to-point ETL pipelines; integration ticket TTR & Project duration for cross-domain dashboards and ML models \\
Unit of analysis & Domain-quarter panel & Org-year cross-section & Project-level \\
Data source & API gateways; data catalogs (Alation, Collibra) & Git repos; ITSM platforms (ServiceNow, Jira) & PMO records; project management systems \\
Ident.\ challenge & Capable domains may receive more autonomy (reverse causality) & Isolating data ATD from code-level debt & Selection: silver-layer investors may differ systematically \\
Preferred design & IV (exogenous reorg) or panel FE & $TD = a + bN^2$ vs.\ $TD = a + bN$ test & Diff-in-diff around silver-layer investment events \\
\bottomrule
\end{tabularx}
\begin{minipage}{\textwidth}
\vspace{4pt}
\footnotesize \textit{Notes.} TTR~=~time-to-resolution. FE~=~fixed effects. IV~=~instrumental variables. Each prediction maps from a specific model mechanism to an operationalized empirical design.
\end{minipage}
\end{table}

\paragraph{Prediction 4: Narrow products proliferate in high-$\alpha$ domains.} \textit{Mechanism:} High-$\alpha$ domains optimize locally because the private synergy benefit $\alpha_i \beta$ may exceed $\kappa / q_i^*$ while still being far below the social optimum. \textit{Outcome variable:} Product generality at the domain level. \textit{Proxy:} Metadata compliance scores. \textit{Unit:} Domain-quarter panel. \textit{Data source:} Data catalog telemetry. \textit{Identification challenge:} High-$\alpha$ domains may also be high-capability, confounding generality with maturity. \textit{Preferred design:} Within-firm variation, controlling for domain size and technology stack.

\paragraph{Empirical entry point.} Of the four predictions, Prediction~3 (silver-layer investment reduces time-to-value) is likely the most tractable first empirical test. Organizations that make discrete, observable investments in centralized silver-layer infrastructure create natural ``treatment events'' amenable to difference-in-differences designs. The outcome variable (project duration for cross-domain analytics) is routinely tracked in project management systems and does not require novel measurement instruments. This design most closely mirrors the model's comparison of governance regimes (Section~\ref{sec:governance}) and would provide an initial empirical anchor for the broader research program. Prediction~2 (quadratic debt scaling) is also testable with existing ITSM data, though the specification test ($N^2$ vs.\ $N$) requires sufficient cross-sectional variation in the number of autonomous domains.

%=============================================================
% 9. DISCUSSION, LIMITATIONS, AND CONCLUSION
%=============================================================
\section{Discussion, Limitations, and Conclusion}\label{sec:discussion}

\subsection{Summary of Results}

Within the model, pure data mesh yields systematic underinvestment in generality. The gap grows with $N$ and $\lambda$; under plausible parameters, a corner solution obtains (Corollary~\ref{cor:trap}). Technical debt grows quadratically. Governance interventions that internalize externalities improve welfare within the model.

\subsection{Theoretical Implications}

The primary contribution is providing microeconomic foundations for a phenomenon widely observed but previously discussed only through heuristics. Formalizing the externality structure identifies the conditions under which underinvestment arises and the mechanisms through which governance can correct it. The reframing of ATD as an economic externality connects data governance to broader programs in platform economics and public goods.

The dynamics of Propositions~\ref{prop:underinvestment} and~\ref{prop:welfareloss} hold regardless of whether interacting nodes are internal departments or independent enterprises, suggesting potential generalizability to inter-organizational data spaces \citep{Duch-Brown2017}---though such extension would require modeling of inter-firm contracting frictions.

\subsection{Practical Implications}

Absent mechanisms that internalize cross-domain benefits, reusable silver-layer assets are unlikely to emerge from domain incentives alone. Either centralized investment or incentive mechanisms appear necessary, but the choice depends on $\lambda$, standardization costs, political feasibility, and tooling availability.

\subsection{Limitations and Scope Conditions}\label{sec:discussion:limits}

\paragraph{Information assumptions.} The model assumes complete information and costless verifiability (Assumption~1). Moral hazard and adverse selection arise when generality is unverifiable; extending the model to hidden action is a natural next step.

\paragraph{Symmetric domains.} Identical parameters are assumed (Assumption~2). Appendix~\ref{app:heterogeneous} shows the result survives---and may worsen---under asymmetry, but a full heterogeneous treatment remains for future work.

\paragraph{Static game.} The core equilibrium is one-shot (Assumption~3). A fully dynamic model with learning, reputation, and norm evolution would enrich the theory.

\paragraph{Reduced-form generality.} The scalar $g_i$ collapses a multidimensional phenomenon. Decomposing generality into component dimensions and modeling their interactions is a promising direction.

\paragraph{Empirical validation.} The calibration is illustrative. Future work should exploit exogenous variation---regulatory mandates, post-merger integration, platform migrations---to identify causal effects.

\paragraph{Organizational complexity.} The model abstracts from bounded rationality, political dynamics, and informal governance. Translating analytical models to complex sociotechnical systems is inherently contingent on institutional detail.

\subsection{Conclusion}

This paper identifies a specific microeconomic mechanism---a positive externality in the provision of cross-domain generality---that prevents efficient data product outcomes under pure decentralization. The generality gap, the data mesh trap corner solution, and the quadratic scaling of technical debt all emerge as analytical consequences of this externality structure, not as empirical claims about any particular organization. Whether a given organization falls within the trap is an empirical question; the model provides the formal framework for answering it.

The paper does not imply that decentralized data ownership is undesirable per se. Rather, it identifies the incentive conditions under which decentralization requires complementary coordination mechanisms. Governance interventions can in principle achieve efficient outcomes, but their effectiveness depends on the organization's ability to observe, verify, and reward genuine generality---a challenge that connects this work to broader research programs on mechanism design in decentralized systems and that defines the empirical agenda for future work.

%=============================================================
% APPENDICES
%=============================================================
\appendix

\section{Proofs}\label{app:proofs}

\subsection{Proof of Proposition~\ref{prop:underinvestment}}\label{app:proof1}

\textbf{Claim:} $g_i^{NE} < g_i^{SO}$ with gap $\Delta g_i = (N-1)\lambda / \gamma_g$.

\begin{proof}
Domain $i$'s private FOC:
\[
\frac{\partial \pi_i}{\partial g_i} = \alpha_i \beta q_i - \gamma_g g_i q_i - \kappa = 0
\]
Social planner's FOC:
\[
\frac{\partial W}{\partial g_i} = \alpha_i \beta q_i + (N-1)\lambda q_i + M\bar{\omega}_i q_i - \gamma_g g_i q_i - \kappa = 0
\]
Solving:
\begin{align*}
g_i^{NE} &= \frac{\alpha_i \beta - \kappa / q_i}{\gamma_g} \\[6pt]
g_i^{SO} &= \frac{\alpha_i \beta + (N-1)\lambda + M\bar{\omega}_i - \kappa / q_i}{\gamma_g}
\end{align*}
Abstracting from the consumer term:
\[
\Delta g_i = \frac{(N-1)\lambda}{\gamma_g} > 0 \qedhere
\]
\end{proof}

\subsection{Proof of Proposition~\ref{prop:debt}}\label{app:proof3}

\textbf{Claim:} $TD_{\text{total}} = \tau \cdot \bar{q} \cdot N(N-1) \cdot \bar{P}$ under corner solution.

\begin{proof}
When $g_i^{NE} = 0$ for all $i$:
\[
TD_{\text{total}} = \tau \sum_{i=1}^{N} \sum_{j \neq i} q_j \cdot P_{ij}
\]
With symmetry ($q_j = \bar{q}$, $P_{ij} = \bar{P}$):
\[
TD_{\text{total}} = \tau \cdot \bar{q} \cdot \bar{P} \cdot N(N-1) \qedhere
\]
\end{proof}

\section{Heterogeneous Domains Extension}\label{app:heterogeneous}

The baseline model assumes symmetric domains. We relax this by allowing heterogeneous analytics value ($\alpha_i$) and asymmetric cross-domain value ($\lambda_{ji}$, denoting the value domain $j$ derives from domain $i$'s generality).

\begin{definition}[Hub domain]
Domain $i$ is a \emph{hub domain} if $\sum_{j \neq i} \lambda_{ji} \gg \bar{\lambda}$, where $\bar{\lambda}$ is the mean cross-domain value. Hub domains produce data consumed heavily across the organization (e.g., ERP, CRM, supply chain systems).
\end{definition}

Under heterogeneity, domain $i$'s equilibrium generality remains:
\[
g_i^{NE} = \max\left\{0, \; \frac{\alpha_i \beta - \kappa / q_i^*}{\gamma_g}\right\}
\]
while the social optimum becomes:
\[
g_i^{SO} = \max\left\{0, \; \frac{\alpha_i \beta + \sum_{j \neq i} \lambda_{ji} - \kappa / q_i^*}{\gamma_g}\right\}
\]

\begin{proposition}[Hub Domain Vulnerability]\label{prop:hub}
Under heterogeneous cross-domain value, the generality gap for domain $i$ is:
\begin{equation}\label{eq:heterogap}
    \Delta g_i = \frac{\sum_{j \neq i} \lambda_{ji}}{\gamma_g}
\end{equation}
Hub domains (high $\sum_{j \neq i} \lambda_{ji}$) exhibit the largest generality gaps, and low-$\alpha_i$ hub domains face the strongest version of the underinvestment problem.
\end{proposition}

\paragraph{Result.} The gap depends on the sum of values that \emph{other domains} place on domain $i$'s data---a quantity that does not enter domain $i$'s private objective.

\paragraph{Economic intuition.} A domain that produces highly valuable data for others but derives little value from its own analytics has the largest gap and the least private incentive to close it. Core operational domains (ERP, supply chain) that function as organizational data hubs face the strongest externality precisely because their data is most valuable to others. This suggests that the underinvestment result is not an artifact of symmetry but a structural feature of decentralized governance, potentially exacerbated by realistic asymmetry.

\paragraph{Scope conditions.} The result assumes that $\lambda_{ji}$ values are common knowledge. Under private information about cross-domain value, additional mechanism design considerations arise.

%=============================================================
% REFERENCES
%=============================================================
\bibliographystyle{apalike}

\end{document}